\documentclass[letterpaper, 10pt, conference]{ieeeconf}      % Use this line for a4
\IEEEoverridecommandlockouts                              % This command is only
\usepackage{amsmath,amssymb,amsfonts}

\usepackage{amsthm}
\usepackage{graphicx}
\usepackage{cite}
\usepackage{xcolor}
\usepackage{algorithmic}
\usepackage{algorithm}
\usepackage[hidelinks]{hyperref}
\usepackage{lipsum}
\usepackage{svg}
\usepackage{comment}
\usepackage{url}
\newtheorem{theorem}{Theorem}
\newtheorem{lemma}{Lemma}
\newtheorem{proposition}{Proposition}
\newtheorem{definition}{Definition}
\newtheorem{assumption}{Assumption}
\newtheorem{corollary}{Corollary}
\newtheorem{remark}{Remark}

\newcommand{\R}{\mathbb{R}}
\newcommand{\C}{\mathbb{C}}

\DeclareMathOperator{\col}{col}
\DeclareMathOperator{\rank}{rank}

\DeclareMathOperator{\im}{im}

\title{\LARGE \bf
Data-Driven Unknown Input Reconstruction for MIMO Systems \\ with Convergence Guarantees
}

\author{Enno Breukelman, Takumi Shinohara, Joowon Lee, and Henrik Sandberg% <-this % stops a space
\thanks{This work was supported in part by the Wallenberg AI, Autonomous Systems and Software Program (WASP) funded by the Knut and Alice Wallenberg Foundation, and by the Swedish Research Council (grant 2023-04770).}% <-this % stops a space
\thanks{The authors are with KTH Royal Institute of Technology, School of Electrical Engineering and Computer Science,
Department of Decision and Control Systems, Malvinas väg 10, SE-100 44 Stockholm, Sweden.
\texttt{\{cebre, tashin, joowon, hsan\}@kth.se}}%%
}

\begin{document}

\maketitle
\thispagestyle{empty}
\pagestyle{empty}

%%%%%%%%%%%%%%%%%%%%%%%%%%%%%%%%%%%%%%%%%%%%%%%%%%%%%%%%%%%%%%%%%%%%%%%%%%%%%%%%
\begin{abstract}
% Existing data-driven inversion methods for MIMO systems either assume known inputs up to the prediction window or lack convergence guarantees tied to system-theoretic properties. 

%%%%% (Joowon) %%%%%
% This paper addresses data-driven reconstruction of unknown inputs to linear time-invariant (LTI) systems.
% Specifically, we propose a novel input estimator based on a constrained least-squares formulation over data Hankel matrices, splitting ...
% Unlike existing methods, our method does not require knowledge of the initial input trajectory
% and is applicable to multi-input multi-output (MIMO) systems as well.
% We show that the proposed estimator is stable if and only if all the invariant zeros of the system lie strictly inside the unit circle.
% This mirrors ...
% (numerical examples)
%%%%%%%%%%%%%%%%%%%%

In this paper, we consider data-driven reconstruction of unknown inputs for deterministic, linear time-invariant (LTI) multiple-input multiple-output (MIMO) systems. 
We propose a novel autoregressive estimator based on a constrained least-squares formulation over Hankel matrices, splitting the problem into an output-consistency constraint and an input-history-matching objective. 
Our method relies on previously recorded input–output data to represent the system, but does not require knowledge of the true input to initialize the algorithm.
We show that the proposed estimator is strictly stable if and only if all the invariant zeros of the trajectory-generating system lie strictly inside the unit circle, which can be verified purely from input and output data.
This mirrors existing results from model-based input reconstruction and closes the gap between model-based and data-driven settings. 
Lastly, we provide numerical examples to demonstrate the theoretical results.
\end{abstract}

\section{Introduction}
\label{sec:Intro}

Data-driven representations of linear time-invariant (LTI) systems have become an established topic, driven by developments in behavioral systems theory \cite{willems_note_2004} and subspace identification methods \cite{van_overschee_subspace_1996}. A central idea in this framework is that the input-output behavior of a dynamical system is represented not through state-space matrices or transfer functions, but through linear combinations of previously recorded input-output trajectories. This perspective has also significantly influenced controller design, enabling closed-loop synthesis without the explicit identification of system matrices \cite{markovsky_data-driven_2008, de_persis_formulas_2020, berberich_robust_2020}.

These data-driven formulations rely on Hankel matrices constructed from measured trajectories. 
Under the assumption of persistently excited inputs, strong connections can be drawn between model-based and behavioral approaches \cite{DBCbook}. 
This includes the ability to identify model properties from data, even when the inputs are not persistently exciting \cite{van_waarde_informativity_2023}.
Despite these emerging connections, some techniques remain rigorously proven only in the model-based setting, with their data-driven counterparts yet to be shown.

One example is the problem of left inversion, namely, the reconstruction of the unique input that generated a measured output. This particular problem is motivated by fault detection, cyber-attack estimation, and inversion-based control, and has foundational work in \cite{willsky_invertibility_1974, moylan_stable_1977}. Recent research continues to connect system properties with the ability to recover unknown inputs, such as \cite{bejarano_unknown_2009, loreto_strong_2023} in the continuous case. 
Model-based inversion and unknown input estimation methods for the discrete-time case are also well documented; see~\cite{gillijns_kalman_2007, sundaram_fault-tolerant_2012}. 
Key aspects of inversion methods are the system's initial condition as well as the existence and position of invariant zeros. 

Motivated by the rise of behavioral and data-driven methods, a growing body of work investigates left inversion directly from data matrices, bypassing explicit model identification. Data-driven estimation approaches that include state estimation in the presence of unknown inputs can be found in \cite{disaro_delayed_2024, disaro_equivalence_2025, mishra_data-driven_2025, turan_data-driven_2022}, while \cite{mishra_data-driven_2023, eun_data-driven_2023, lee_input-output_2025, shi_data-driven_2022} specifically address system invertibility and input reconstruction. These methods share a common structure: a previously recorded input-output trajectory is used to represent the system dynamics, while online measurements are used to estimate the current input. A key limitation of existing approaches is that, in addition to online output measurements, they require exact knowledge of the true inputs immediately preceding the estimation window, or at initialization. This assumption is unrealistic in practice, since it is precisely these inputs that one seeks to recover.

For the deterministic, i.e. noise-free, single-input single-output (SISO) case, \cite{lee_input-output_2025} address this limitation and provide a convergence proof for arbitrary initial conditions using transfer-function arguments. For multiple-input multiple-output (MIMO) systems, \cite{shi_data-driven_2022} propose an estimator by characterizing a general matrix inverse; however, they do so without providing a system-theoretic explanation of when and why the method succeeds. Consequently, an open problem remains for MIMO systems. No existing data-driven method guarantees convergence of the input reconstruction for arbitrary initial conditions while connecting these guarantees to the system-theoretic properties that govern model-based inversion.

Our main contributions are as follows:
\begin{enumerate}
\item We develop a novel, data-driven, autoregressive input estimation algorithm for MIMO systems that does not require knowledge of the initial input trajectory. 
\item We provide a rigorous convergence analysis linking the estimator's behavior to the invariant zeros of the trajectory-generating system, mirroring the structural assumptions of model-based inversion.
%if all invariant zeros are stable, the algorithm converges for any unknown initial input trajectory; if none are present, the correct input is recovered exactly. These guarantees 
\item We establish a necessary and sufficient condition which can be checked from input and output data, for all invariant zeros to lie strictly inside the unit circle.
\item We illustrate the theoretical results in numerical examples, highlighting the role of invariant-zero locations.
\end{enumerate}

\emph{Outline:}
In Section~\ref{sec:Prelim}, we review existing results on model-based and data-driven input reconstruction with known initial conditions. Section~\ref{sec:Input_Reconstr_unknown_IC} then presents our algorithm for estimating unknown inputs under arbitrary initial conditions, along with its convergence analysis. Numerical examples demonstrate our theoretical findings in Section~\ref{sec:NumericalExamples}, and we conclude the paper in Section~\ref{sec:Concl}.

\emph{Notation:}
$\mathbb{R}$ denotes the set of real numbers, $\mathbb{C}$ the complex numbers, and $\mathbb{N}$ the natural numbers; column vectors in $\mathbb{R}^n$ are denoted by small letters, e.g., $x$, and matrices in $\mathbb{R}^{n \times m}$ by capital letters, e.g., $A$. 
%With $\col(u, y)$ we denote the vertical concatenation of the two column vectors $u, y$. 
We use the notations $ \| x \|_2 $ and $ \| A \|_2 $ to denote the $ \ell_2 $ norm of a vector $ x $ and the induced $ \ell_2 $ norm of a matrix, respectively. By $\im(A)$, we denote the image, i.e., the column space of $A$, and by $\dim(\im(A))$, we denote the dimension of said image.
For a square matrix $A$, the set of eigenvalues (spectrum) and the maximum absolute eigenvalue (spectral radius) are represented by $\sigma(A)$ and $\rho(A)$, respectively.

\section{Preliminaries}
\label{sec:Prelim}

In this section, we collect existing results on model-based and data-driven system inversion.

\subsection{Model-based System Inversion}

Consider the discrete-time linear time-invariant system $\boldsymbol{S}$ with $u \in \R^m$, $y \in \R^p$ and $x \in \R^n$ satisfying
\begin{equation}\label{eq:SS_S}
  \boldsymbol{S}: \; \left\{\begin{aligned} x_{k+1} &= A x_{k} + B u_k, \\
  y_k &= C x_k + D u_k, \end{aligned} \right.
\end{equation}
where $A \in \R^{n \times n}$, $B \in \R^{n \times m}$, $C \in \R^{p \times n}$ and $D \in \R^{p \times m}$.
%The input-output behavior of $\boldsymbol{S}$ can also be analyzed using the steady-state transfer matrix
%\begin{equation}
%  G_S(z) = C(A - zI)^{-1}B + D.
%\end{equation}
%We are interested in finding an $L$-delay inverse, for which it holds that
%\begin{equation}
%  G_\boldsymbol{S}^{\text{inv}}(z)\, G_S(z) = \frac{1}{z^L}\, I_m,
%\end{equation}
%with the delay $L\leq n$, c.f. Section 3.2.1 in \cite{gillijns_kalman_2007}.
%
To design the inverse of $\boldsymbol{S}$, we begin by considering the stacked vector of $L+1$ outputs
\begin{equation}
  y_{k:k+L} = \begin{bmatrix} y_k^\top & y_{k+1}^\top & \cdots & y_{k+L}^\top \end{bmatrix}^\top,
\end{equation}
which can also be written as
\begin{equation}\label{eq:stacked_outputs}
  y_{k:k+L} = \mathcal{O}_L\, x_k + \mathcal{I}_L\, u_{k:k+L},
\end{equation}
where $\mathcal{O}_L$ denotes the \emph{observability matrix} of order $L$ and $\mathcal{I}_L$ the \emph{invertibility matrix} for the unknown input of order $L$.
These matrices are computed recursively using
\begin{equation*}
\mathcal{O}_L = \begin{bmatrix} C \\ \mathcal{O}_{L-1} A \end{bmatrix}, \quad
\mathcal{I}_L = \begin{bmatrix} D & 0 \\ \mathcal{O}_{L-1} B & \mathcal{I}_{L-1} \end{bmatrix},
\end{equation*}
with $\mathcal{O}_0 = C$ and $\mathcal{I}_0 = D$.
% Suppose $x_k=0$ such that $ y_{k:k+L} = \mathcal{I}_L\, u_{k:k+L}$. 
We characterize the existence of a delayed left inverse, i.e. the invertibility property, as the ability to recover $u_k$ uniquely from $y_{k:k+L}$ for some $L \in \mathbb{N}$, given a zero or known initial state $x_k$.
%Note that, in the assumed minimality case, $\mathcal{O}_n$ is of full column rank.
%To estimate the inputs from \eqref{eq:stacked_outputs}, we require an $L$-delay invertibility from system $\boldsymbol{S}$ in \eqref{eq:SS_S}.

\begin{lemma}[Chapter 2 in \cite{sundaram_fault-tolerant_2012}]\label{lemma:invertibility}
  The system $\boldsymbol{S}$ is invertible if and only if either of the following conditions holds:
  \begin{enumerate}
    \item[i)] There exists at least one $z \in \C$ for which
    \begin{equation*}
      %\rank[\boldsymbol{M}(z)] = 
      \rank\begin{bmatrix} A - zI & B \\ C & D \end{bmatrix} = n + m.
    \end{equation*}
    \item[ii)] There exists $L\leq n$ and $P \in \R^{m \times (L+1)p}$ for which
    \begin{equation*}
      P\, \mathcal{I}_L = \begin{bmatrix} I_m & 0 \end{bmatrix}.
    \end{equation*}
    %\item[iii)] There exists $L\leq n$ for which 
    %\begin{equation*}
    %  \rank[\mathcal{I}_L] = m + \rank[\mathcal{I}_{L-1}].
    %\end{equation*}
  \end{enumerate}
\end{lemma}
%
%In the remainder of the paper, we refer to the existence of an $L$-delay left inverse as the invertibility property. 
Note that invertibility requires that $p \geq m$, i.e., there are at least as many outputs as inputs.
We now make the following assumptions on $\boldsymbol{S}$.
\begin{assumption}\label{ass:minimal}
  The system $\boldsymbol{S}$  
  \begin{enumerate}
      \item is minimal, i.e., $(A, C)$ is observable and $(A, B)$ is controllable, and
      \item is invertible for some $L_0 \leq n$, c.f.  Lemma~\ref{lemma:invertibility}.
      %\item has no overlapping poles and zeros of $\boldsymbol{S}$ in the multivariate case.
  \end{enumerate} 
\end{assumption}
Knowing that $\boldsymbol{S}$ is invertible, we invoke statement ii) from Lemma~\ref{lemma:invertibility} with $L = L_0$ and left-multiply \eqref{eq:stacked_outputs} by $P$ to obtain
\begin{equation}\label{eq:inputs_inverse}
  u_k = P\, y_{k:k+L} - P\, \mathcal{O}_L\, x_k.
\end{equation}
Next, we substitute \eqref{eq:inputs_inverse} back into \eqref{eq:SS_S} providing us with a dynamical system that has $y_{k:k+L}$ as inputs and $u_k$ as outputs.
%\begin{equation}
%  \boldsymbol{S}^{\text{inv}} : \begin{cases} x_{k+1} &= (A - BP\mathcal{O}_L)\, x_k + BP\, y_{k:k+L}, \\ u_k &= -P\mathcal{O}_L\, x_k + P\, y_{k:k+L}. \end{cases}
%\end{equation}
Using $\tilde{A} := (A - BP\mathcal{O}_L)$, $\tilde{B} := BP$, $\tilde{C} := -P\mathcal{O}_L$ and $\tilde{D} := P$, we can write the inverse $\boldsymbol{S}^{\text{inv}}$ as
\begin{equation}\label{eq:SS_inv}
  \boldsymbol{S}^{\text{inv}}: \; \left\{ \begin{aligned} x_{k+1} &= \tilde{A}\, x_k + \tilde{B}\, y_{k:k+L}, \\ u_k &= \tilde{C}\, x_k + \tilde{D}\, y_{k:k+L}. \end{aligned} \right.
\end{equation}
Note that $\boldsymbol{S}$ and $\boldsymbol{S}^{\text{inv}}$ describe the exact same dynamical system with the same states $x_k$.
The matrix $P$ from statement ii) in Lemma~\ref{lemma:invertibility} is not unique, rendering the matrix $\tilde{A}$ not unique.
Before we can connect the poles of $\tilde{A}$ to the dynamics of $\boldsymbol{S}$, we introduce the concept of \emph{invariant zeros}.
\begin{definition}[Section 4.5.1 in \cite{skogestad_multivariable_2010}]\label{def:inv_zeros}
The invariant zeros of $\boldsymbol{S}$ are described by the values of $z \in \C$ for which %the matrix pencil $\mathbf{M}(z)$ loses column rank, i.e.,
\begin{align*}
  %\rank[\mathbf{M}(z)] = 
  \rank\begin{bmatrix} A - z I & B \\ C & D \end{bmatrix} < n + m.
\end{align*}
\end{definition}
Using this definition, we state the following lemma.
\begin{lemma}\label{lemma:eigen_A_tilde}
  All eigenvalues of $\tilde{A}$ can be placed freely by design of the matrix $P$, except for the invariant zeros of $\boldsymbol{S}$.
\end{lemma}
\begin{proof}
  See the proof of Theorem 3.2 in \cite{sundaram_fault-tolerant_2012}. 
\end{proof}

We now return to the input estimation using the system inverse $\boldsymbol{S}^{\text{inv}}$. The input $u_k$ can be computed directly from $y_{k:k+L}$ in \eqref{eq:inputs_inverse} if an exact state estimate $x_k$ is available, or if the matrix $P$ can be designed such that $-P\mathcal{O}_L = \tilde{C} = 0$.

If neither is the case, it is straightforward to see that 
\begin{align}
    \hat{u}_k - u_k &= \tilde{C} (\hat{x}_k - x_k),\\
    \hat{x}_{k+1} - x_{k+1} & = \tilde{A}(\hat{x}_k - x_k),
\end{align}
where $\hat{u}_k$ and $\hat{x}_k$ follow $\boldsymbol{S}^\text{inv}$, but with a different initial state. 
Therefore, we need to consider the dynamics of the inverse $\boldsymbol{S}^{\text{inv}}$ in \eqref{eq:SS_inv}, governed by the poles in $\tilde{A}$.
This shows that for the model-based input reconstruction, the convergence for arbitrary initial conditions depends on the existence and location of invariant zeros of $\boldsymbol{S}$.
In particular, Theorems 2.5 and 2.8 in \cite{sundaram_fault-tolerant_2012} show that
\begin{equation}
\rank\begin{bmatrix} \mathcal{O}_L & \mathcal{I}_L \end{bmatrix} = n + \rank[\mathcal{I}_L], \text{ for some } L \leq n,
\end{equation}
if and only if there are no invariant zeros. Under this condition, a matrix $P$ satisfying statement ii) in Lemma~\ref{lemma:invertibility} and $P\mathcal{O}_L = 0$ exists. 
We summarize the results from model-based input reconstruction in the following proposition.
\begin{proposition}\label{prop:model_based}
  Suppose that Assumption~\ref{ass:minimal} holds, and let the matrix $P$ be suitably designed according to statement ii) in Lemma~\ref{lemma:invertibility}. Given only the output trajectory $y_{k:k+L}$, the following statements hold for model-based input reconstruction with an arbitrary initial condition $\hat{x}_0 \in \mathbb{R}^n$:
  \begin{enumerate}
    \item[i)] $\hat{u}_k = u_k$ for all $k \geq \nu$, with $\nu \leq n$, if and only if every invariant zero of $\boldsymbol{S}$ lies at the origin; here $\nu = 0$ is attainable if and only if $\boldsymbol{S}$ has no invariant zero, and
    \item[ii)] $\hat{u}_k - u_k \to 0$ as $k \to \infty$ if and only if $\boldsymbol{S}$ has only stable invariant zeros.
  \end{enumerate}
\end{proposition}

Statement i) follows from Lemma~\ref{lemma:eigen_A_tilde}: zeros at the origin render $\tilde A$ nilpotent for a suitable $P$, so any initial error is vanishes within at
most $n$ steps, whereas a nonzero invariant zero becomes an eigenvalue of $\tilde A$ regardless of $P$. The special case $\nu=0$ requires $\tilde C = 0$, which removes the state dependence in \eqref{eq:inputs_inverse} altogether and is possible
exactly when $\boldsymbol{S}$ has no invariant zero.

Note that if the initial condition is exact, $\hat x_0=x_0$, due to invertibility, the input estimate will satisfy $\hat{u}_k = u_k$ for any set of invariant zeros. Additionally, note that we refer to stability in the Schur sense. 
\begin{remark}\label{rem:strong_det}
In the model-based case, the property that all invariant zeros lie inside the unit circle is also called \emph{strong detectability}, while having no invariant zeros is referred to as \emph{strong observability} \cite{sundaram_fault-tolerant_2012}. %This nomenclature is primarily used in the context of state estimation. 
\end{remark}

\subsection{Data-Driven System Inversion}

We begin by introducing standard notations in the data-driven literature. For a signal $u_{0:T+N+L} \in \R^{m(T+N+L+1)}$, we denote the block Hankel matrix $\mathcal{H}_t(u_{0:T+N+L}) \in \R^{mt \times (T+N+L+2-t)}$ of depth $t$ by
\begin{equation}
  \mathcal{H}_t(u_{0:T+N+L}) := \begin{bmatrix} u_0 & u_{1} & \cdots & u_{T+N+L+1-t} \\ u_{1} & u_{2} & \cdots & u_{T+N+L+2-t} \\ \vdots & \vdots & \ddots & \vdots \\ u_{t-1} & u_{t} & \cdots & u_{T+N+L} \end{bmatrix}.
\end{equation}
Note that the subscript $t$ refers to the number of block rows of the Hankel matrix.
Then, $u_{0:T+N+L}$ is said to be \emph{persistently exciting} of order $t$ if the Hankel matrix $\mathcal{H}_t(u_{0:T+N+L})$ has full row rank \cite{willems_note_2004}.

To be able to state the system inversion problem for trajectories generated by the system $\boldsymbol{S}$ in \eqref{eq:SS_S}, let ${u^d := u_{0:T+N+L}}$ and $y^d := y_{0:T+N+L}$ denote recorded data, and partition the Hankel matrices as follows
\begin{align}
  \begin{bmatrix} U_p \\ U_f^L \end{bmatrix} &:= \mathcal{H}_{N+L+1}(u^d) \in \R^{m(N+L+1)\times (T+1)}, \\
  \begin{bmatrix} Y_p \\ Y_f^L \end{bmatrix}& := \mathcal{H}_{N+L+1}(y^d)\in \R^{p(N+L+1)\times (T+1)},
\end{align}
where $U_p$ and $U_f^L$ consist of the first $mN$ rows and the last $m(L+1)$ rows of $\mathcal{H}_{N+L+1}(u^d)$, respectively.
Similarly, $Y_p$ and $Y_f^L$ consist of the first $pN$ rows and the last $p(L+1)$ block rows of $ \mathcal{H}_{N+L+1}(y^d)$, respectively.
We define $U_f$ as the first $m$ rows of $U_f^L$.
We refer to these matrices as \emph{offline data}, which we use to represent the system $\boldsymbol{S}$. Note that $N\geq n$ and $L\geq L_0$, where $n$ is the state dimension and $L_0$ the inherent delay of the trajectory generating system $\boldsymbol{S}$. 

While $n$ and $L_0$ may not be known a~priori, there exist methods to estimate them from data. For example, see Section~2.2 in \cite{van_overschee_subspace_1996} to compute $n$, and Algorithm~2 in \cite{mishra_data-driven_2025} to compute $L_0$.
%In other words
%To be able to assure richness of information of these collected trajectories, we assume the following.
%
\begin{assumption}\label{assume:PE}
    The signal $u^d$ is persistently exciting of order $n+ N +L+1$, i.e., the corresponding Hankel matrix $\mathcal{H}_{n+N+L+1}(u^d)$ has full row rank.
\end{assumption}

We are now in place to state the data-based inversion problem for a window of one time step.
\begin{lemma}[Theorem 2 in \cite{eun_data-driven_2023}]\label{lem:invertible_dd_inv}
    Suppose that Assumptions~\ref{ass:minimal} and~\ref{assume:PE} hold.
    %Also, assume that $N\geq n$ and $u^d$ is persistently exciting of order $n+N+L+1$.
    Then, $u_k = \hat{u}_k = U_f g$, where $g \in \R^{T+1}$ is a solution to
    \begin{align}\label{eq:H^L}
    H^L g = \begin{bmatrix} U_p \\ Y_p \\ Y_f^L \end{bmatrix} g = \begin{bmatrix} u_{k-N:k-1} \\ y_{k-N:k-1} \\ y_{k:k+L} \end{bmatrix} 
    \begin{array}{@{}l@{}}
  \leftarrow \text{ini. cond. for } u, \\
  \leftarrow \text{ini. cond. for } y,\\
  \leftarrow \text{measured output.}
  \end{array}
    \end{align}
\end{lemma}

The key insight is that the invertibility of $\boldsymbol{S}$ results in $U_f$ being linearly dependent on $U_p$, $Y_p$, and $Y_f^L$. 
%In other words, the rows in $U_f$ live in the space spanned by the rows of $H^L$.
%\begin{equation}\label{eq:nullspace_equality}
%  \mathcal{N}(H^L) = \mathcal{N}\!\left(\begin{bmatrix} U_p \\ Y_p \\ Y_f^L \end{bmatrix}\right) = \mathcal{N}\!\left(\begin{bmatrix} U_p \\ U_f \\ Y_p \\ Y_f^L \end{bmatrix}\right).
%\end{equation}
We can use this to write $U_f = h^\top H^L$, where $h \in \R^{(N(m+p) + p(L+1)) \times m}$.
If we combine this with \eqref{eq:H^L}, we obtain
\begin{align}\label{eq:input_rec_h}
  \hat{u}_k = U_f g = h^\top H^L g = h^\top \begin{bmatrix} u_{k-N:k-1} \\ y_{k-N:k-1} \\ y_{k:k+L} \end{bmatrix}.
\end{align}
For an invertible system, the estimate satisfies $\hat{u}_k = u_k$ whenever the true input $u_{k-N:k-1}$ is available, either from direct knowledge at every step $k$ or from exact initialization of a recursive algorithm.

\section{Data-Driven Input Reconstruction with Unknown Initial Condition}
\label{sec:Input_Reconstr_unknown_IC}
If the initial condition $u_{k-N:k-1}$ is not known exactly, % the algorithm prescribed in 
\eqref{eq:input_rec_h} is no longer well defined. 
This is because, contrary to \eqref{eq:H^L}, for $\hat{u}_{k-N:k-1} \neq u_{k-N:k-1}$, the fundamental lemma in \cite{willems_note_2004} does not guarantee the existence of some $g$ for which
\begin{equation}\label{eq:dd_input_recon}
  H^L g = \begin{bmatrix} U_p \\ Y_p \\ Y_f^L \end{bmatrix} g = \begin{bmatrix} \hat{u}_{k-N:k-1} \\ y_{k-N:k-1} \\ y_{k:k+L} \end{bmatrix}.
\end{equation}
Since the computation of $\hat{u}_k$ depends recursively on previous, 
potentially incorrect, estimates $\hat{u}_{k-N:k-1}$, a stability 
analysis of the estimation is necessary. 

In the SISO case, \cite{lee_input-output_2025} investigated the convergence properties of \eqref{eq:input_rec_h} when $h^\top = U_f (H^L)^\dagger$ and $\hat{u}_{k-N:k-1}$ unknown, using the Moore--Penrose inverse.
Their result in Corollary~4 states that the unknown input estimate converges if and only if the underlying LTI system is minimum-phase.
For the MIMO case, \cite{shi_data-driven_2022} develops a recursive algorithm using a generalized inverse of $H^L$ via linear matrix inequalities (LMIs), yielding a stable estimator. However, existence conditions of such an estimator are not provided.
The main contribution of this paper is the design of a novel unknown input estimator for the MIMO case, whose convergence can be characterized by invariant zeros of $\boldsymbol{S}$.

\subsection{Data-Driven Input Estimator Design}
Instead of computing the Moore--Penrose inverse $h^\top = U_f (H^L)^\dagger$, we split the problem of computing a candidate $g$ into two parts.
By virtue of the fundamental lemma, we know that there must exist a $g$ for which
\begin{equation}
  Y g = \begin{bmatrix} Y_p \\ Y_f^L \end{bmatrix} g = \begin{bmatrix} y_{k-N:k-1} \\ y_{k:k+L} \end{bmatrix}.
\end{equation}
This is independent of the unknown inputs and has to hold strictly.
In a second part, we propose to minimize the error in $U_p g = \hat{u}_{k-N:k-1}$, which is the mismatch between the previous input estimates and the feasible linear combinations of the offline data.
This procedure can be summarized in a least-squares problem over Hankel matrices, where the output consistency is enforced as a hard constraint
\begin{align}\label{eq:constrained_QP}
  g^\star &\in \arg\min_g \|U_p g - \hat{u}_{k-N:k-1}\|_2^2 \nonumber \\
  &\quad \text{subject to } \begin{bmatrix} Y_p \\ Y_f^L \end{bmatrix} g = \begin{bmatrix} y_{k-N:k-1} \\ y_{k:k+L} \end{bmatrix}, \nonumber \\
  \hat{u}_k &= U_f g^\star.
\end{align}
If the true input $\hat{u}_{k-N:k-1} = u_{k-N:k-1}$ is used in \eqref{eq:constrained_QP}, there exists some $g$ that results in a zero residual, i.e., ${\|U_p g^\star - \hat u_{k-N:k-1}\|^2_2 = 0}$. This means that there exists a $g$ satisfying \eqref{eq:input_rec_h}.
%\begin{equation}
%    H^L g = \begin{bmatrix} u_{k-N:k-1} \\ y_{k-N:k-1} \\ y_{k:k+L} \end{bmatrix}.
%\end{equation}
If we additionally assume that the system is invertible, by the proof in \cite{eun_data-driven_2023}, $\hat{u}_k = U_f g^\star$ is the unique and exact solution $u_k$.

To facilitate further analysis, we choose the following closed-form solution~$g^\star$ to \eqref{eq:constrained_QP}:
%For that, let 
%\begin{equation}
% \hat{u} = u_{k-N:k-1}, \quad y = \begin{bmatrix} y_{k-N:k-1} \\ y_{k:k+L} \end{bmatrix}.
%\end{equation}
\begin{align*}
g^\star = Y^\dagger \begin{bmatrix} y_{k-N:k-1} \\ y_{k:k+L} \end{bmatrix} &+ (I - Y^\dagger Y)(U_p(I - Y^\dagger Y))^\dagger  \\
&\times\left(\hat u_{k-N:k-1} - U_p Y^\dagger \begin{bmatrix} y_{k-N:k-1} \\ y_{k:k+L} \end{bmatrix}\right).
\end{align*}
The matrix $Y^\dagger$ denotes the Moore–Penrose pseudoinverse of $Y$, so that $\Pi_Y^\perp := (I - Y^\dagger Y)$ is the orthogonal projector onto $\ker (Y)$.
To complete the algorithm, we compute the input estimate $\hat{u}_k = U_f g^\star$ as
\begin{align}
  \hat{u}_k  & =    U_f \Pi_Y^\perp (U_p \Pi_Y^\perp)^\dagger \hat{u}_{k-N:k-1}  \nonumber\\
   & \quad + U_f (Y^\dagger - \Pi_Y^\perp (U_p \Pi_Y^\perp)^\dagger U_p Y^\dagger) y_{k-N:k+L} \label{eq:closed_form_solution}\\
     & =  {} M_u\, \hat{u}_{k-N:k-1} + M_y\, y_{k-N:k+L}, \label{eq:abbreviated_algorithm}
\end{align}
using 
\begin{align}
    M_u &:= U_f \Pi_Y^\perp (U_p \Pi_Y^\perp)^\dagger, \label{eq:M_u}\\
    M_y &:= U_f (Y^\dagger - \Pi_Y^\perp (U_p \Pi_Y^\perp)^\dagger U_p Y^\dagger). \label{eq:M_y}
\end{align}

%In the case that $\hat{u}_{k-N:k-1} = u_{k-N:k-1}$, we know that \eqref{eq:abbreviated_algorithm} converges in one step, i.e., it will directly yield $u_k$.

\subsection{Data-driven Input Estimator Convergence Analysis}
A necessary condition for the data-driven input estimator to converge for arbitrary initial conditions is that $\boldsymbol{S}$ is invertible, i.e., $\hat{u}_k = u_k = M_u\, u_{k-N:k-1} + M_y\, y_{k-N:k+L}$. 
%\begin{assumption}\label{ass:invertible}
%  The system $\boldsymbol{S}$ is invertible.
%\end{assumption}
To be able to investigate the convergence of \eqref{eq:abbreviated_algorithm} towards the correct input, we hence compute the difference between the two iterations as
\begin{align}
  e_k = \hat{u}_k - u_k ={}& M_u(\hat{u}_{k-N:k-1} - u_{k-N:k-1}) \nonumber \\
  & + M_y(y_{k-N:k+L} - y_{k-N:k+L}) \nonumber \\
  ={}& M_u\, e_{k-N:k-1}.
\end{align}
By introducing $\epsilon_k := e_{k-N:k-1}$ we can analyze the convergence as a linear system following $\epsilon_{k+1} = R\, \epsilon_k$ for some $R$ admitting the structure
\begin{equation*}
  R := \left[\begin{array}{ccccc}
  0 & I_m &  \cdots & 0 &0 \\
  %0 & 0 & I_m & \cdots & 0 \\
  \vdots & & \ddots & &\vdots \\
  0 & 0 &  \cdots & 0 & I_m \\
  \multicolumn{5}{c}{M_u}
  \end{array}\right].
\end{equation*}
We are now ready to present the main result.
\begin{theorem}\label{thm:R_stability}
Under Assumptions~\ref{ass:minimal} and~\ref{assume:PE}, $R$ is Schur stable if and only if system $\boldsymbol{S}$ in \eqref{eq:SS_S} has only stable invariant zeros.
\end{theorem}

We present the proof of Theorem~\ref{thm:R_stability} at the end of this section.
Theorem~\ref{thm:R_stability} shows that the estimation is asymptotically correct for every initialization if and only if $\boldsymbol{S}$ has only stable invariant zeros, consistent with the stable pole-zero cancellations in the SISO case~\cite{lee_input-output_2025}.
\begin{comment}
    \begin{corollary}\label{col:strongly_detectable}
  Under Assumptions~\ref{ass:minimal} and~\ref{assume:PE}, the system $\boldsymbol{S}$ has only stable invariant zeros if and only if $\rho(R)<1$.  
\end{corollary}
\begin{proof}
    It follows directly from Theorem~\ref{thm:R_stability}.
\end{proof}
\end{comment}

As mentioned in Remark~\ref{rem:strong_det}, the property that $\boldsymbol{S}$ has only stable invariant zeros is, in the model-based case, referred to as strong detectability.
A data-driven condition for strong detectability was already presented in Theorem 5.7 in \cite{DBCbook}, but it requires state, input, and output data. 
In contrast, $\rho(R)$ can be computed only using input-output data, as visible from the definition of $M_u$ in~\eqref{eq:M_u}.

By Theorem~\ref{thm:R_stability}, the convergence property of the proposed algorithm in \eqref{eq:abbreviated_algorithm} is characterized as follows.

% \textcolor{red}{Add some intuitions of the theorem and explanations on the relationship to the data-driven strongly detectable.}

\begin{corollary}\label{col:asymptotic_convergence}
Under Assumptions~\ref{ass:minimal} and~\ref{assume:PE}, the following statements hold for the data-driven unknown-input estimation algorithm in \eqref{eq:abbreviated_algorithm}, for an arbitrary initial estimate $\hat{u}_{k-N:k-1} \in \mathbb{R}^{m N}$:
\begin{enumerate}
\item[(i)] $\hat{u}_k = u_k$ from the first iteration if and only if every invariant zero of $\boldsymbol{S}$ lies at the origin, the case of no invariant zeros included, and
\item[(ii)] $\hat{u}_k - u_k \to 0$ as $k \to \infty$ if and only if $\boldsymbol{S}$ has only stable invariant zeros.
% \item[(iii)] $\hat{u}_k - u_k \to \infty$ as $k \to \infty$ if and only if $\boldsymbol{S}$ has at least one unstable invariant zero.
\end{enumerate}
\end{corollary}

Corollary~\ref{col:asymptotic_convergence} admits a direct interpretation: if $\boldsymbol{S}$ has no invariant zeros or invariant zeros only at the origin, the input is recovered after a single update. 
If $\boldsymbol{S}$ has only stable invariant zeros, the effect of an incorrect initialization decays asymptotically.
This establishes a direct connection between the data-driven algorithm in \eqref{eq:abbreviated_algorithm} and the model-based Proposition~\ref{prop:model_based}: the system-theoretic conditions guaranteeing convergence under inexact initialization are identical in both settings, differing only in the number of steps required for exact recovery. 
The data-driven estimator achieves it within its first iteration, because that iteration already spans $N \geq n$ samples, whereas the model-based inverse needs up to $n$ steps to annihilate the initial error.

We now continue with untangling $R$ using \eqref{eq:M_u}, leading up to the proof of our main result and the corresponding corollary. For that, we find relationships between the data matrices $U_p, U_f, Y_p$ and $Y_f^L$.
Inspired by the derivations in \cite{eun_data-driven_2023}, we can find the following connection using the state-space matrices of an inverse $\boldsymbol{S}^{\text{inv}}$:
\begin{align}\label{eq:extended_inverse_data}
  \begin{bmatrix} U_p \\ U_f \end{bmatrix} ={}& 
  %\begin{bmatrix} \tilde{C} \\ \tilde{C}\tilde{A} \\ \vdots \\ \tilde{C}\tilde{A}^{N} \end{bmatrix} 
  \tilde{\mathcal{O}}_{N}
  X + 
  %\underbrace{\begin{bmatrix} \tilde{D} & 0 & \cdots & 0 \\ \tilde{C}\tilde{B} & \tilde{D} & \cdots & 0 \\ \vdots & \ddots & \ddots & \vdots \\ \tilde{C}\tilde{A}^{N-1}\tilde{B} & \tilde{C}\tilde{A}^{N-2}\tilde{B} & \cdots & \tilde{D} \end{bmatrix}}_{
  \tilde{\mathcal{I}}_{N} W \begin{bmatrix} Y_p \\ Y_f^L \end{bmatrix},
\end{align}
where $\tilde{\mathcal{O}}_{N}$ and $\tilde{\mathcal{I}}_{N}$ are defined analogously to $\mathcal{O}_N$ and $\mathcal{I}_N$, using the system matrices of $\boldsymbol{S}^{\text{inv}}$ in \eqref{eq:SS_inv}.
The matrix $W \in \mathbb{R}^{(N+1)p(L+1)\times p(N+L+1)}$ consists of $N+1$ stacked, horizontally shifted, identity matrices of size $p(L+1)$, each shifted by $p$ columns. 
The matrix $X \in \mathbb{R}^{n\times (T+1)}$ contains state trajectories and is not available. 
% In the following, we exploit the relationship among the data matrices $X, U_p, U_f, Y_p, Y_f^L$.

Recall that $Y\Pi_Y^\perp = Y - YY^\dagger Y = 0$. Then, right-multiplying \eqref{eq:extended_inverse_data} by $\Pi_Y^\perp$ results in 
\begin{align}
    \begin{bmatrix}
        U_p \\ U_f
    \end{bmatrix}  \Pi_Y^\perp = \tilde{\mathcal{O}}_N X \Pi_Y^\perp, \text{ where } \tilde{\mathcal{O}}_{N}
  = \begin{bmatrix} \tilde{\mathcal{O}}_{N-1} \\ \tilde{C}\tilde{A}^N \end{bmatrix}.
\end{align}
If we now let $\tilde{X} = X\Pi_Y^\perp$, we get
\begin{align}
U_p \Pi_Y^\perp = \tilde{\mathcal{O}}_{N-1} \tilde{X}, \quad
U_f  \Pi_Y^\perp= \tilde{C}\tilde{A}^N \tilde{X}.
\end{align}
\begin{comment}
    Note that 
\begin{equation}
  %\begin{bmatrix} \tilde{C} \\ \tilde{C}\tilde{A} \\ \vdots \\ \tilde{C}\tilde{A}^{N} \end{bmatrix} 
  \tilde{\mathcal{O}}_{N}
  = \begin{bmatrix} \tilde{\mathcal{O}}_{N-1} \\ \tilde{C}\tilde{A}^N \end{bmatrix}, \quad \tilde{\mathcal{I}}_{N} = \begin{bmatrix} \overline{\mathcal{I}} \\ \underline{\mathcal{I}} \end{bmatrix},
\end{equation}
where $\overline{\mathcal{I}}$ and $\underline{\mathcal{I}}$ partition the invertibility matrix $\tilde{\mathcal{I}}_{N+1}$ along the last $m$ rows.
%
%Now, recall that
%\begin{equation}
%  M_u = U_f \Pi_Y^\perp (U_p \Pi_Y^\perp)^\dagger.
%\end{equation}
We recall that $Y\Pi_Y^\perp = Y - YY^\dagger Y = 0$ and then substitute
\begin{align}
U_p &= \tilde{\mathcal{O}}_{N-1} X + \overline{\mathcal{I}} W Y, \\
U_f &= \tilde{C}\tilde{A}^N X + \underline{\mathcal{I}} WY
\end{align}
into \eqref{eq:M_u} to obtain
\begin{align*}
 M_u &= (\tilde{C}\tilde{A}^N X + \underline{\mathcal{I}}W Y)\Pi_Y^\perp \left((\tilde{\mathcal{O}}_{N-1} X + \overline{\mathcal{I}} W Y) \Pi_Y^\perp\right)^\dagger \\
  &= (\tilde{C}\tilde{A}^N X \Pi_Y^\perp)(\tilde{\mathcal{O}}_{N-1} X \Pi_Y^\perp)^\dagger.
\end{align*}
If we now let $\tilde{X} = X\Pi_Y^\perp$, we get
\end{comment}
We can use this to write $M_u$ as
\begin{equation}
  M_u = \tilde{C}\tilde{A}^N \tilde{X} (\tilde{\mathcal{O}}_{N-1} \tilde{X})^\dagger.
\end{equation}
Then, the following lemma characterizes the image of $\tilde X$.

\begin{lemma}\label{lemma:range_X_tilde}
Suppose that Assumptions~\ref{ass:minimal} and~\ref{assume:PE} hold.
%and $u^d$ is persistently exciting of order $N + n + 1 + L$ with $N \geq n$.
Let
\begin{align}
    \mathcal{V}_0 := \{\xi \in \R^n&:\exists \bar u \in \R^{m(N+L+1)}~\nonumber \\
    &\mathrm{s.t}.~\mathcal{O}_{N+L} \xi + \mathcal{I}_{N+L} \bar u = 0\}.
\end{align}
Then, $\im(\tilde X) = \mathcal{V}_0$.
\end{lemma}
\begin{proof}
  Since $\im(\Pi_Y^\perp) = \ker(Y)$, we have $\im(\tilde X)=\{Xw:w\in \ker(Y)\}$.
  Let $\xi = Xw$ with $w \in \ker(Y)$.
  Consider now the stacked outputs of the forward system
  \begin{equation}
    Y= \mathcal{O}_{N+L} X + \mathcal{I}_{N+L} \underbrace{\begin{bmatrix} U_p \\ U_f^L \end{bmatrix}}_U.
  \end{equation}
  %where $U_f^L$ is constructed in the same way as $Y_f^L$ based on the input data $u^d$.
  Due to persistence of excitation, $U$ has full row rank. 
  We then find
  \begin{align}
    Yw = 0 = \mathcal{O}_{N+L} Xw + \mathcal{I}_{N+L} Uw
    %\iff \mathcal{O}_{N+L+1} Xw = \mathcal{O}_{N+L+1} \xi = -\mathcal{I}_{N+L+1} Uw,
  \end{align}
    which implies $\xi \in \mathcal{V}_0$, and thus we have $\im(\tilde X) \subseteq \mathcal{V}_0$.

  Conversely, let $\xi \in \mathcal{V}_0$.
  Then there exists an input sequence $\bar u \in \R^{m(N+L+1)}$ such that $\mathcal{O}_{N+L} \xi + \mathcal{I}_{N+L} \bar u = 0$.
  Therefore, $(\bar u, 0)$ is a $(N+L+1)$-length input-output trajectory of the system \eqref{eq:SS_S}.
  By the fundamental lemma, %Lemma~\ref{lemma:trajectory_of_S}, 
  applied with horizon $N+L+1$, there exists $w \in \R^{T+1}$ such that $Uw = \bar u$ and $Yw = 0$.
  Hence,
  \begin{align*}
      0 \!=\! Yw \!=\! \mathcal{O}_{N+L} Xw + \mathcal{I}_{N+L}\bar u \!=\! \mathcal{O}_{N+L}(Xw - \xi).
  \end{align*}
  Since $(A,C)$ is observable and $N+L \geq n -1$, $\mathcal{O}_{N+L}$ has full column rank (i.e., $\ker(\mathcal{O}_{N+L})=\{0\}$), and thus $Xw = \xi$.
  Due to $Yw = 0$, we have $w \in \ker(Y)$, and therefore
  \begin{align}
      \xi = Xw \in X\ker(Y) = \im(X \Pi_Y^\perp)= \im(\tilde X),
  \end{align}
  which proves $\im(\tilde X) = \mathcal{V}_0$. 
  %
  %  Recall that an invariant zero $z_i \in \C$ with corresponding $\xi_i, \mu_i \neq 0$ satisfies \begin{equation}\label{eq:inv_zero}
  %  \begin{bmatrix} A - z_i I & B \\ C & D \end{bmatrix} \begin{bmatrix} \xi_i \\ \mu_i \end{bmatrix} = \begin{bmatrix} 0 \\ 0 \end{bmatrix}.
  %\end{equation} 
  %  Successively right-multiplying by $z_i$ yields trajectories $x_k = z_i^k \xi_i$, $u_k = z_i^k \mu_i$ with $y_k \equiv 0$, which over $N+L+1$ steps gives
  %  \begin{equation}
  %      \mathcal{O}_{N+L+1}\xi_i + \mathcal{I}_{N+L+1} \begin{bmatrix} \mu_i \\ z_i \mu_i \\ z_i^2 \mu_i \\ \vdots \end{bmatrix} = 0,
  %  \end{equation}
  %  confirming that $\im(\tilde{X})$ characterizes the output-nulling states of $\boldsymbol{S}$. 
\end{proof}

Under Assumptions~\ref{ass:minimal} and~\ref{assume:PE}, $\mathcal{V}_0$ is the largest weakly unobservable subspace (see Chapter 7 in \cite{trentelman2001ControlTheoryLinear} for details) of $\boldsymbol{S}$ in \eqref{eq:SS_S}.
Hence, for $\xi \in \mathcal{V}_0 = \im(\tilde X)$, there exists an input sequence such that the corresponding output is zero.
%The following lemma shows that $\im(\tilde X)$ is $\tilde A$-invariant.
\begin{lemma}\label{lemma:tilde_A_invariant}
Under Assumptions~\ref{ass:minimal} and~\ref{assume:PE}, $\im(\tilde X)$ is $\tilde A$-invariant, i.e., $\tilde A \im(\tilde X) \subseteq \im(\tilde X)$.
\end{lemma}
\begin{proof}
Let $\xi \in \im(\tilde X)$.
From Lemma~\ref{lemma:range_X_tilde}, $\im(\tilde X) = \mathcal{V}_0$, which is the largest weakly unobservable subspace, and hence there exists a zero-output trajectory $(x_k,u_k)_{k\geq 0}$ such that $y_k \equiv 0$ of the system $\boldsymbol{S}$ from $x_0 = \xi$.
From Lemma~\ref{lemma:invertibility}, there exists a matrix $P$ such that $P\mathcal{I}_L = [~I_m~~0~]$.
From \eqref{eq:stacked_outputs}, we have $\mathcal{O}_L x_k + \mathcal{I}_L u_{k:k+L} = 0 $, and left-multiplying by $P$ yields $u_k = -P\mathcal{O}_L x_k = \tilde C x_k$.
Hence, 
\begin{align*}
    x_{k+1} = Ax_{k} + Bu_{k} = Ax_k - BP\mathcal{O}_L x_k = \tilde A x_k.
\end{align*}
Since the trajectory is output-nulling, $x_{k+1} \in \mathcal{V}_0$.
In particular, for $k=0$, we obtain $\tilde A \xi = x_1 \in \mathcal{V}_0 = \im(\tilde X)$, which implies $\tilde A \im(\tilde X) \subseteq \im(\tilde X)$.
\end{proof}

To continue investigating the properties of $M_u$, we state the following technical lemma on $\tilde{\mathcal{O}}_{N-1}$.

\begin{lemma}\label{lemma:full_rank_o}
Under Assumptions~\ref{ass:minimal} and~\ref{assume:PE}, the observability matrix $\tilde{\mathcal{O}}_{N-1}$ of the inverse system $\boldsymbol{S}^{\text{inv}}$ in \eqref{eq:SS_inv} %is injective on $\im(\tilde{X})$, i.e., 
satisfies $\im(\tilde X)\cap \ker(\tilde{\mathcal{O}}_{N-1}) = \{0\}$.
\end{lemma}
\begin{proof}
  We prove this by contradiction.
  Consider a nonzero vector $v \in \im(\tilde X)\cap \ker(\tilde{\mathcal{O}}_{N-1})$.
  By Lemma~\ref{lemma:range_X_tilde}, $\im(\tilde X)=\mathcal{V}_0$, and hence there exists a zero-output trajectory of the system $\boldsymbol{S}$ starting from $x_0 = v$.
  Since the inverse system $\boldsymbol{S}^{\text{inv}}$ in \eqref{eq:SS_inv} describes the same trajectory of $\boldsymbol{S}$, and $y_k \equiv 0$ along the trajectory, we have $x_{k+1}=\tilde A x_k$ and $u_k = \tilde C x_k$.
  Therefore, $x_k = \tilde A^k v$ and $u_k = \tilde C \tilde A^k v$.
  By the premise that $v \in \ker (\tilde{\mathcal{O}}_{N-1})$, it follows that $u_k = \tilde C \tilde A^k v = 0$ for $k=0,\ldots,N-1$.
  
  For these time steps, the original system $\boldsymbol{S}$ reduces to $x_{k+1}=Ax_{k}$ and $y_k=Cx_k$, so that $x_k = A^k v$ for $k=0,\ldots,N-1$.
  Since the output trajectory is zero, we obtain $CA^k v = 0$ for $k=0,\ldots,N-1$, implying $v \in \ker (\mathcal{O}_{N-1})$.
  The assumptions that $(A,C)$ is observable and $N\geq n$ imply $\ker(\mathcal{O}_{N-1}) = \{0\} $, and thus $v = 0$, which contradicts the premise.
  Therefore, $\im(\tilde X)\cap \ker(\tilde{\mathcal{O}}_{N-1}) = \{0\}$. 
\end{proof}

%Since we assumed that there are no overlapping poles and zeros in the multivariate case in Assumption \ref{ass:minimal}, Lemma~\ref{lemma:full_rank_o} guarantees that $\tilde{\mathcal{O}}_{N-1}$ is injective for all $v \in \im(\tilde{X})$.
This allows us to proceed with the last technical lemma leading to the proof of Theorem~\ref{thm:R_stability}.
 
\begin{lemma}\label{lemma:R_eigvals}
Under Assumptions~\ref{ass:minimal} and~\ref{assume:PE}, $R\,\tilde{\mathcal{O}}_{N-1} v = \tilde{\mathcal{O}}_{N-1} \tilde{A} v$ for $v \in \im(\tilde{X})$.
\end{lemma}
\begin{proof}
Consider $v \in \im(\tilde X)$.
There exists $\beta \in \mathbb{R}^{T+1}$ such that $v = \tilde X \beta$.
Due to the design of $R$ and $\tilde{\mathcal{O}}_{N-1}$, we get
\begin{equation}
  R\,\tilde{\mathcal{O}}_{N-1} v = \begin{bmatrix} \tilde{C}\tilde{A}v \\ \tilde{C}\tilde{A}^2 v \\ \vdots \\ \tilde{C}\tilde{A}^N \tilde{X}(\tilde{\mathcal{O}}_{N-1} \tilde{X})^\dagger \tilde{\mathcal{O}}_{N-1} v \end{bmatrix}.
\end{equation}
The last block can be written as 
\begin{multline*}
    \tilde{C}\tilde{A}^N \tilde{X}(\tilde{\mathcal{O}}_{N-1} \tilde{X})^\dagger \tilde{\mathcal{O}}_{N-1} v  \\ =\tilde{C}\tilde{A}^N \tilde{X}(\tilde{\mathcal{O}}_{N-1} \tilde{X})^\dagger(\tilde{\mathcal{O}}_{N-1} \tilde X) \beta.
\end{multline*}
Let $w := \tilde{X}(\tilde{\mathcal{O}}_{N-1} \tilde{X})^\dagger(\tilde{\mathcal{O}}_{N-1} \tilde X) \beta$.
Note that $w \in \im (\tilde X)$.
Then,
\begin{multline}
    \tilde{\mathcal{O}}_{N-1}w = \tilde{\mathcal{O}}_{N-1}\tilde{X}(\tilde{\mathcal{O}}_{N-1} \tilde{X})^\dagger(\tilde{\mathcal{O}}_{N-1} \tilde X) \beta  \\ = \tilde{\mathcal{O}}_{N-1}\tilde{X} \beta = \tilde{\mathcal{O}}_{N-1} v.
\end{multline}
By virtue of Lemma~\ref{lemma:full_rank_o}, we have ${\im(\tilde X)\cap \ker(\tilde{\mathcal{O}}_{N-1}) = \{0\}}$, %$\tilde{\mathcal{O}}_{N-1}$ is injective on $\im(\tilde X)$,
providing us with $w=v$, and hence $v =\tilde{X}(\tilde{\mathcal{O}}_{N-1} \tilde{X})^\dagger\tilde{\mathcal{O}}_{N-1} v $.
Therefore,
\begin{equation}
  R\,\tilde{\mathcal{O}}_{N-1} v = \begin{bmatrix} \tilde{C}\tilde{A}v \\ \tilde{C}\tilde{A}^2 v \\ \vdots \\ \tilde{C}\tilde{A}^N v \end{bmatrix} = \tilde{\mathcal{O}}_{N-1} \tilde{A} v,
\end{equation}
which concludes the proof.
\end{proof}

Now, we are in place to proof Theorem~\ref{thm:R_stability}.

\begin{proof}[Proof of Theorem~\ref{thm:R_stability}]
Choose a basis matrix $X_z \in \mathbb{R}^{n \times r}$ of $\im(\tilde X)$, where $r := \dim(\im(\tilde X))$.
From Lemma~\ref{lemma:tilde_A_invariant}, $\im(\tilde X)$ is $\tilde A$-invariant, and thus we can define $A_z \in \mathbb{R}^{r \times r}$ such that
\begin{align}\label{eq:A_z}
\tilde A X_z = X_z A_z.
\end{align}
Define $J := \tilde{\mathcal{O}}_{N-1} X_z$.
By Lemma~\ref{lemma:full_rank_o}, $\tilde{\mathcal{O}}_{N-1}$ is injective on $\im(\tilde X)$, and hence $J$ has full column rank.
According to Lemma~\ref{lemma:R_eigvals}, we have $R\,\tilde{\mathcal{O}}_{N-1} v = \tilde{\mathcal{O}}_{N-1} \tilde{A} v$ for $v \in \im(\tilde{X})$.
Applying this with $v = X_z \alpha$ yields
\begin{multline*}
RJ\alpha
= R \tilde{\mathcal{O}}_{N-1} X_z \alpha  \\ = \tilde{\mathcal{O}}_{N-1} \tilde A X_z \alpha = \tilde{\mathcal{O}}_{N-1} X_z A_z \alpha = JA_z \alpha
\end{multline*}
for all $\alpha \in \R^r$, and therefore,
\begin{align}\label{eq:RJ}
RJ = J A_z,
\end{align}
which implies $\im(J)$ is $R$-invariant.

Consider an orthogonal matrix $Q := [~Q_1~~Q_2~]$ such that the columns of $Q_1$ and $Q_2$ form an orthonormal basis of $\im(J)$ and $\im(J)^\bot$, respectively.
Then, we have
\begin{align}
Q^\top R Q =
\begin{bmatrix}
Q_1^\top R Q_1 & Q_1^\top R Q_2 \\
Q_2^\top R Q_1 & Q_2^\top R Q_2
\end{bmatrix},
\end{align}
where $\im(R Q_1) \subseteq \im(J)$, and thus $Q_2^\top R Q_1 = 0$.
Consequently, it follows that
\begin{align}\label{eq:QRQ_relation}
Q^\top R Q =
\begin{bmatrix}
Q_1^\top R Q_1 & Q_1^\top R Q_2 \\
0 & Q_2^\top R Q_2
\end{bmatrix},
\end{align}
which implies that $R$ is similar to the matrix on the right-hand side of \eqref{eq:QRQ_relation}.
Since the spectrum of a block upper triangular matrix is the union of the spectra of its diagonal blocks,
\begin{align}
\sigma(R) = \sigma(Q_1^\top R Q_1) \cup \sigma(Q_2^\top R Q_2).
\end{align}

We next address $\sigma(Q_1^\top R Q_1)$.
Since $J$ has full column rank, we may choose $Q_1 = J(J^\top J)^{-1/2}$.
Then,
\begin{align}
Q_1^\top R Q_1 &= (J^\top J)^{-1/2}J^\top R J (J^\top J)^{-1/2} \nonumber \\
&\overset{(\ref{eq:RJ})}{=} (J^\top J)^{-1/2}J^\top J A_z (J^\top J)^{-1/2} \nonumber \\
& = (J^\top J)^{1/2} A_z (J^\top J)^{-1/2}.
\end{align}
Note that $(J^\top J)^{1/2}$ and $(J^\top J)^{-1/2}$ are invertible because $J$ has full column rank.
Therefore, $Q_1^\top R Q_1$ and $A_z$ are similar, and thus $\sigma(Q_1^\top R Q_1) = \sigma(A_z)$.

Next, consider $\sigma(Q_2^\top R Q_2)$.
By the definition of $R$, we can rewrite it as
\begin{align*}
R =
\underbrace{
\begin{bmatrix}
0 & I & & \\
& &   \!\ddots\! \\
& &  &   I \\
0 & 0 & \!\cdots\! & 0
\end{bmatrix}}_{\tilde S} + (e_N \otimes I) M_u,
\end{align*}
where $e_N$ is the $N$th canonical basis vector of $\R^N$ and $\otimes$ denotes the Kronecker product.
Then, we have
\begin{multline*}
M_u(I\!-\!J J^\dagger) \! \\ =\! \tilde C \tilde A^N \tilde X (\tilde{\mathcal{O}}_{N-1} \tilde X)^\dagger\!\left(I\!-\!(\tilde{\mathcal{O}}_{N-1} \tilde X)(\tilde{\mathcal{O}}_{N-1} \tilde X)^\dagger\right)\!=\!0.
\end{multline*}
Moreover, since the columns of $Q_2$ lie in $\im(J)^\bot$, we have $JJ^\dagger Q_2 = 0$, and hence $(I-JJ^\dagger)Q_2 = Q_2$.
Consequently, we have
\begin{align}
RQ_2 = R(I-JJ^\dagger)Q_2 &= (\tilde S + (e_N \otimes I)M_u)(I-JJ^\dagger)Q_2 \nonumber \\
&= \tilde S(I-JJ^\dagger)Q_2 = \tilde S Q_2,
\end{align}
which implies that $Q_2^\top R Q_2 = Q_2^\top \tilde S Q_2$ and $\sigma(Q_2^\top R Q_2) = \sigma(Q_2^\top \tilde S Q_2)$.
Together with the above result, the spectrum of $R$ can be characterized by
\begin{align} \label{eq:R_spectrum}
\sigma(R) = \sigma(A_z) \cup \sigma(Q_2^\top \tilde S Q_2).
\end{align}
From the definition of $\tilde S$, $\| \tilde S \|_2 = 1$.
Thus, we have
\begin{align*}
\rho(Q_2^\top \tilde S Q_2) \leq \|Q_2^\top \tilde S Q_2\|_2 \leq  \|Q_2^\top\|_2  \|S\|_2  \|Q_2 \|_2 =  1.
\end{align*}
To exclude eigenvalues on the unit circle, assume for contradiction that $(Q_2^\top \tilde S Q_2)v = \lambda v$ for some nonzero vector and $|\lambda|=1$.
Then,
\begin{align*}
\|v\|_2 = \| (Q_2^\top \tilde S Q_2) v\|_2 \overset{(a)}{\leq} \|\tilde S Q_2 v\|_2 \leq \|Q_2 v\|_2 = \|v\|_2,
\end{align*}
so equality holds throughout.
Equality $(a)$ implies $\tilde S Q_2 v \in \im(J)^\bot$.
Therefore,
\begin{align}
\tilde S Q_2 v = Q_2 Q_2^\top \tilde S Q_2 v = \lambda Q_2 v,
\end{align}
which implies $\lambda$ is an eigenvalue of $\tilde S$.
However, $\tilde S$ is nilpotent, so its only eigenvalue is $0$, which is a contradiction.
Thus, $(Q_2^\top \tilde S Q_2)$ has no eigenvalue on the unit circle. Therefore, $\rho(Q_2^\top \tilde S Q_2) < 1$, and from \eqref{eq:R_spectrum},
\begin{align}\label{eq:R_radius}
\rho(R) < 1 \Longleftrightarrow \rho(A_z) < 1.
\end{align}

Finally, we show that $\sigma(A_z)$ is identical to the set of invariant zeros of \eqref{eq:SS_S}.
By Lemma~\ref{lemma:range_X_tilde}, $\im(\tilde X) = \mathcal{V}_0$, and hence $X_z$ is a basis matrix of $\mathcal{V}_0$.
For any $\xi \in \mathcal{V}_0$, if $\bar u := u_{0:L}$ is an output-nulling input sequence satisfying $\mathcal{O}_L \xi + \mathcal{I}_L \bar u = 0$, then, since there exists a matrix $P$ such that $P\mathcal{I}_L = [~I_m~~0~]$ from Lemma~\ref{lemma:invertibility}, its first element is uniquely given by $u_0 = -P\mathcal{O}_L \xi = \tilde C \xi$.
Therefore, the map $\tilde A \xi = (A-BP\mathcal{O}_L)\xi = A\xi + B \tilde C \xi = A\xi + Bu_0$ is precisely the zero-dynamics map on $\mathcal{V}_0$.

Let $\lambda \in \sigma(A_z)$.
Then, there exists $\alpha \neq 0$ such that $A_z \alpha = \lambda \alpha$.
Set $\xi := X_z \alpha \neq 0$ and $\mu:=\tilde C \xi$.
Since $\xi \in\mathcal{V}_0$, the input $\mu$ is the first input of a zero-output trajectory, and thus $C\xi + D\mu = 0$.
Moreover, we obtain
\begin{align*}
    (A-\lambda I)\xi +B\mu &=(A+B\tilde C-\lambda I)\xi=(\tilde A-\lambda I)\xi \\
    &= (\tilde A X_z - X_z A_z)\alpha \overset{(\ref{eq:A_z})}{=} 0.
\end{align*}
Hence,
\begin{align}\label{eq:invariant_zero_thm}
\begin{bmatrix} A - \lambda I & B \\ C & D
\end{bmatrix}
\begin{bmatrix} \xi \\ \mu
\end{bmatrix} =
\begin{bmatrix} 0 \\ 0
\end{bmatrix}.
\end{align}
Therefore, $\lambda$ is an invariant zero of \eqref{eq:SS_S}.

Conversely, let $\lambda$ be an invariant zero of \eqref{eq:SS_S}.
Then, by definition, there exists a nonzero pair $(\xi, \mu)$ such that \eqref{eq:invariant_zero_thm} holds.
The trajectory defined by $x_k = \lambda^k \xi$ and $u_k = \lambda^k \mu$ yields $y_k \equiv 0$, and therefore $\xi \in \mathcal{V}_0 = \im(\tilde X) = \im(X_z)$.
Considering $k=0$, from (\ref{eq:SS_inv}), we have $\mu = \tilde C \xi + \tilde D y_{0:L} = \tilde C \xi$.
By the definition of $\tilde C$, thus, $\mu = -P \mathcal{O}_L \xi$.
Also, $\xi = X_z \alpha $ for some $\alpha \neq 0$.
Consequently, it follows that
\begin{align*}
\tilde A \xi = (A-BP\mathcal{O}_L)\xi = A\xi + B\mu = \lambda \xi.
\end{align*}
Using (\ref{eq:A_z}) and $\xi = X_z \alpha$, we obtain
\begin{align*}
X_z A_z \alpha = \tilde A X_z \alpha = \lambda X_z \alpha.
\end{align*}
Since $X_z$ has full column rank, left-multiplying by $X_z^\dagger$ gives $A_z \alpha = \lambda \alpha$.
Therefore, $\lambda \in \sigma(A_z)$.
Hence, the set of invariant zeros of \eqref{eq:SS_S} coincides with $\sigma(A_z)$.
Combining this with \eqref{eq:R_radius}, we conclude that $R$ is Schur stable if and only if \eqref{eq:SS_S} has only stable invariant zeros, which proves the claim. 
\end{proof}
Lastly, we provide the proof of Corollary~\ref{col:asymptotic_convergence}.

\begin{proof}[Proof of Corollary~\ref{col:asymptotic_convergence}]
Since $e_k = M_u e_{k-N:k-1}$ and $e_{k-N:k-1}$ ranges over all of $\mathbb{R}^{mN}$,
statement (i) is equivalent to $M_u = 0$. By the proof of Theorem~\ref{thm:R_stability},
$\sigma(A_z)$ is the set of invariant zeros of $\boldsymbol{S}$.
If these all lie at the origin, $A_z$ is nilpotent with $A_z^N = 0$, since
$A_z \in \mathbb{R}^{r \times r}$ with $r \leq n \leq N$. Applying \eqref{eq:A_z}
repeatedly gives $\tilde A^N X_z = X_z A_z^N = 0$, so $\tilde A^N$ vanishes on
$\im(X_z) = \im(\tilde X)$. As every column of $\tilde X$ lies in $\im(\tilde X)$,
we obtain $\tilde A^N \tilde X = 0$ and hence $M_u = 0$. Conversely, $M_u = 0$
renders $R = \tilde S$ nilpotent, so $\sigma(A_z) \subseteq \sigma(R) = \{0\}$
by \eqref{eq:R_spectrum}.
Statement (ii) follows directly from Theorem~\ref{thm:R_stability}.
\end{proof}

\section{Numerical Examples}
\label{sec:NumericalExamples}

We illustrate the theoretical results through numerical examples, applying \eqref{eq:constrained_QP} to invertible systems with various invariant zero configurations.
First, we address numerical considerations for implementing \eqref{eq:closed_form_solution}.

\subsection{Numerical Implementation}

The algorithm in \eqref{eq:closed_form_solution} requires computing $\Pi_Y^\perp = I - Y^\dagger Y$, which can be numerically inconvenient to form explicitly for large $T$.
To compensate for this, we use a nullspace parameterisation that yields an equivalent problem of smaller dimension.
For that, we compute the singular value decomposition (SVD) $Y = U_Y \Sigma_Y V_Y^\top$ and partition $V_Y$ into range-space columns $V_r$ and null-space columns $V_{\text{null}}$.
Using this, every solution to $Yg = y$ can now also be written as
\begin{equation}\label{eq:splitting_Y}
g = Y^\dagger y + V_{\text{null}} \alpha, \quad \alpha \in \R^{T+1 - r_Y},
\end{equation}
where $Y^\dagger y$ is a particular solution and $V_{\text{null}}\alpha$ always lies in $\ker(Y)$. For the practical implementation, we truncate the singular values of $Y$ at $10^{-4}$.
%With this parameterization, no feasible $g$ is excluded.
Substituting \eqref{eq:splitting_Y} into the objective of \eqref{eq:constrained_QP} gives an unconstrained problem
\begin{equation}\label{eq:low_dim_OP}
\min_\alpha \|U_p V_{\text{null}} \alpha - (\hat{u} - U_p Y^\dagger y)\|^2_2,
\end{equation}
which has dimension $T +1 - r_Y$, where $r_Y = \rank[Y]$. We solve \eqref{eq:low_dim_OP} via truncated SVD, with a threshold of $10^{-3}$.
Because the substitution is an exact reparameterisation of the feasible set, the recovered
\begin{equation}
g^\ast = Y^\dagger y + V_{\text{null}} U_{p,0}^\dagger (\hat{u} - U_p Y^\dagger y),
\quad U_{p,0} = U_p V_{\text{null}},
\end{equation}
is the minimizer of the original constrained problem. Finally, we obtain
\begin{align}
  M_u &:= U_f V_{\text{null}} U_{p,0}^\dagger, \\
  M_y &:= U_f (Y^\dagger - M_u U_p Y^\dagger)
\end{align}
for the recursive estimation algorithm\footnote{The code can be accessed at \url{https://github.com/ennobr/DD_INV_MIMO}.}.

\subsection{Numerical examples}
We consider three numerical examples, each with two inputs, two outputs, and four states, differing only in their set of invariant zeros.

\textbf{Stable invariant zeros.} Figure~\ref{fig:stable_zero} shows results for a system with only stable invariant zeros at $z=0.7$ and $z=0.8$. The estimate converges asymptotically to the true input. After estimation begins (marked by the vertical black line), the residual drops sharply over $N$ steps as the algorithm flushes out the incompatible initial guess, then decays asymptotically to zero.
\begin{comment}
    \begin{align*}
A &= \begin{bmatrix}
0.5 & 1 & 0 & 0 \\
0 & 0.6 & 0 & 0 \\
0 & 0 & 0.3 & 1 \\
0 & 0 & 0 & 0.4
\end{bmatrix}, &
B &= \begin{bmatrix}
0 & 0 \\
1 & 0 \\
0 & 0 \\
0 & 1
\end{bmatrix}, \\
C &= \begin{bmatrix}
-1 & 5 & 0 & 0 \\
0 & 0 & -1 & 2
\end{bmatrix}, &
D &= \begin{bmatrix}
0 & 0 \\
0 & 0
\end{bmatrix}.
\end{align*}
\end{comment}
%
\begin{figure}[h]
  \centering
  %\vspace{-3mm}
  \includegraphics[width=\columnwidth]{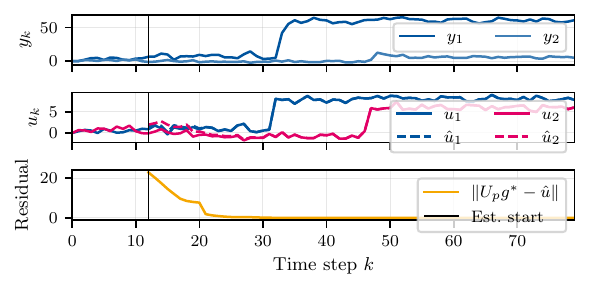}
  % \vspace{-8mm}
  \caption{With $N=10$, estimation converges quickly for a system with only stable invariant zeros.}
  \label{fig:stable_zero}
\end{figure}

\textbf{No invariant zeros.} Figure~\ref{fig:no_zero} presents a strongly observable system with no invariant zeros. The algorithm yields an exact input estimate from the first step. The residual converges to zero in finite time, indicating that the $g$ generated in early iterations is incompatible with the fundamental lemma, yet this does not affect the input estimate.
%
\begin{comment}
    \begin{align*}
A &= \begin{bmatrix}
0.25 & 0.25 & 0 & 0.5 \\
0.5 & 0 & 0 & 0.25 \\
-0.25 & -0.25 & 0 & 0 \\
0.25 & 0.25 & 0 & 0
\end{bmatrix}, &
B& = \begin{bmatrix}
0 & -1 \\
0 & 0 \\
0 & 1 \\
-1 & 0
\end{bmatrix}, \\
C &= \begin{bmatrix}
1 & -1 & 0 & 0 \\
1 & 0 & -1 & 0
\end{bmatrix}, &
D& = \begin{bmatrix}
0 & 0 \\
0 & 0
\end{bmatrix}.
\end{align*}
%
\end{comment}
%
\begin{figure}[h]
  \centering
  %\vspace{-3mm}
  \includegraphics[width=\columnwidth]{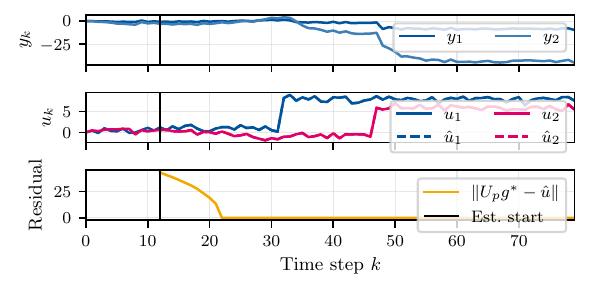}
  % \vspace{-8mm}
  \caption{With $N=10$, estimation is exact from the first step for a strongly observable system.}
  \label{fig:no_zero}
\end{figure}

\textbf{Unstable invariant zeros.} Figure~\ref{fig:unstable_zero} shows a system with an unstable invariant zero at $z=1.25$. The input estimate diverges while the residual converges to zero, confirming that the algorithm finds a $g$ satisfying the fundamental lemma. Yet, this $g$ does not provide a unique input estimate, a consequence of the unstable invariant zero.
%
\begin{comment}
    \begin{align*}
A &= \begin{bmatrix}
0.25 & 0.25 & 0 & 0.5 \\
0.5 & 0 & 0 & 0.25 \\
-0.25 & -0.25 & 0 & 0 \\
0.25 & 0.25 & 0 & 0
\end{bmatrix}, &
B& = \begin{bmatrix}
2 & -2 \\
-1 & 2 \\
2 & 0 \\
0 & -1
\end{bmatrix}, \\
C &= \begin{bmatrix}
2 & 1 & 1 & -2 \\
0 & -1 & -2 & -2
\end{bmatrix}, &
D &= \begin{bmatrix}
0 & 0 \\
0 & 0
\end{bmatrix}.
\end{align*}
\end{comment}
%
\begin{figure}[h]
  \centering
  %\vspace{-3mm}
  \includegraphics[width=\columnwidth]{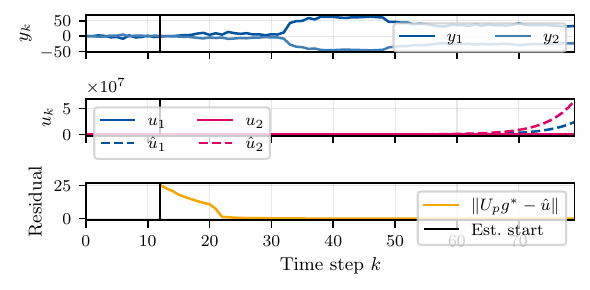}
  %\vspace{-8mm}
  \caption{With $N=10$, estimation diverges despite the algorithm finding a suitable $g$, due to an unstable invariant zero.}
  \label{fig:unstable_zero}
\end{figure}

\section{Conclusions and Future Work}
\label{sec:Concl}
In this paper, we established a rigorous connection between model-based and data-driven input reconstruction for MIMO systems, without requiring knowledge of the initial input trajectory.
The central result is that our proposed data-driven estimator inherits the same convergence conditions as its model-based counterpart: the estimation error decays to zero if and only if all invariant zeros are stable, while invariant zeros confined to the origin, in particular their absence, guarantee exact recovery in a single step.
As a byproduct, we obtained a necessary and sufficient condition for all invariant zeros being stable, verifiable purely from input-output data.
To ensure numerical tractability, we proposed an SVD-based nullspace reparametrization and demonstrated our theoretical findings across three distinct invariant-zero configurations through numerical examples.
Future research directions include extensions to process and measurement noise. 
%and a transfer-function-based approach to MIMO data-driven inversion, building on the SISO results in~\cite{lee_input-output_2025}.

\section*{Usage of Generative AI}
During the preparation of this work, the authors used Claude AI to reason about derivations, improve the syntax and grammar in the manuscript, and support code generation for the numerical examples. 
After using this tool, the authors reviewed and edited the content and take full responsibility for the publication's content.

\bibliographystyle{IEEEtran}
\bibliography{CDC26}

\end{document}